%% file: icmp_v3.tex
\documentclass[11pt,a4paper]{article} 
\usepackage[a4paper,hmargin=2.8cm,vmargin=3cm]{geometry}
\usepackage{amsmath,amsthm,amsfonts,amssymb}
\usepackage{authblk}
   
\usepackage{graphicx}    
 \usepackage{todonotes}
\usepackage[bookmarksnumbered=true]{hyperref} 
\usepackage[capitalise]{cleveref}
\crefname{equation}{}{}
\Crefname{lem}{Lemma}{Lemmas}   
\usepackage{tikz-cd}
\usetikzlibrary{positioning,calc,fadings,decorations.pathreplacing,shadings}
  
\include{macros}


\title{Effective Dynamics of Interacting Fermions from Semiclassical Theory to the Random Phase Approximation} 
\author[]{Niels Benedikter}
\affil[]{Universit\`a degli Studi di Milano, Dipartimento di Matematica, Via Cesare Saldini 50, 20133 Milano, Italy\\ORCID: \href{https://orcid.org/}{0000--0002--1071--6091}, e--mail: \href{mailto:niels.benedikter@unimi.it}{niels.benedikter@unimi.it}}  

\begin{document}
\maketitle

\begin{abstract}
I review results concerning the derivation of effective equations for the dynamics of interacting Fermi gases in a high--density regime of mean--field type. Three levels of effective theories, increasing in precision, can be distinguished: the semiclassical theory given by the Vlasov equation, the mean--field theory given by the Hartree--Fock equation, and the description of the dominant effects of non--trivial entanglement by the random phase approximation. Particular attention is given to the discussion of admissible initial data, and I present an example of a realistic quantum quench that can be approximated by Hartree--Fock dynamics.
\end{abstract}

\tableofcontents

\section{Interacting Fermi Gases at High Density}
Interacting fermions make up much of our world, from metals and semiconductors to neutron stars. Their quantum mechanical description is very complicated because a system of N particles is described by vectors in the (antisymmetrized) $N$--fold tensor product of $L^2(\Rbb^{3})$. As the particle number $N$ is usually huge (easily of the order of $10^{23}$), the Schrödinger equation becomes quickly inaccessible by numerical methods. Effective evolution equations provide a solution: in certain idealized physical regimes they allow an efficient approximation in terms of simpler theories, where ``simpler'' may mean of lower numerical complexity or even explicitly solvable. In this review I present different effective descriptions of the time evolution providing increasing precision of approximation.

\medskip

In this section I introduce the starting point of the quantum mechanical investigation, i.\,e., the fundamental description in terms of the Schrödinger equation. Moreover I discuss the high--density physical regime modelled as a coupled mean--field and semiclassical scaling limit. In the further sections I review, in order of increasing precision of approximation, recent results in the derivation of effective evolution equations. I proceed from the semiclassical approximation (the Vlasov equation) over the mean--field approximation (the Hartree--Fock equation) to the random phase approximation (formulated as bosonization).

\medskip

\paragraph{Schrödinger Equation} The quantum mechanical description is given by the Hamiltonian
\begin{equation} \label{eq:hamiltonian}
 H := - \sum_{i=1}^N \Delta_{x_i} + \lambda \sum_{1 \leq i < j \leq N} V(x_i - x_j)\;, \quad \lambda \in \Rbb\;,
\end{equation}
acting as a self--adjoint operator on $L^2(\Rbb^3)^{\otimes N} \simeq L^2(\Rbb^{3N})$, or more precisely, since we consider fermions, on its antisymmetric subspace; i.\,e., on functions $\psi \in L^2(\Rbb^{3N})$ satisfying
\begin{equation} \label{eq:antisymmetry}
\psi(x_1,x_2,\ldots,x_N) = \sgn(\sigma) \psi(x_{\sigma(1)},x_{\sigma(2)},\ldots,x_{\sigma(N)}) \quad \textnormal{for } \sigma \in \mathcal{S}_N\;.
 \end{equation}
This subspace will be denoted $L^2_\textnormal{a}(\Rbb^{3N})$.
The Hamiltonian generates the dynamics of the system according to the Schrödinger equation: given initial data $\psi_0 \in L^2_\textnormal{a}(\Rbb^{3N})$, the evolution is given by the solution to
\begin{equation}\label{eq:schroedinger}
 i\frac{\di \psi}{\di t}(t) = H \psi(t)\;, \qquad \psi(0) = \psi_0\;.
\end{equation}
If the initial data $\psi_0$ is antisymmetric, so is the solution $\psi(t)$ at all times $t \in \Rbb$.

In this review I discuss the approximation of solutions to \cref{eq:schroedinger} by simpler initial value problems. This of course depends on the choice of initial data, and I will dedicate particular attention to the discussion of the physically most important classes of initial data.

\paragraph{Mean--Field and Semiclassical Scaling Regime}
The Hamiltonian \cref{eq:hamiltonian} describes an extremely wide variety of physical systems, depending on the parameters such as the choice of the interaction potential $V$, of the sign and size of the coupling constant $\lambda$, the density, and the initial data. No approximation can describe all regimes; therefore we impose a specific choice of the parameters. The simplest case are mean--field type scaling regimes: a large number ($N \to \infty$) of particles in a fixed volume (whose size is defined by restricting $\Rbb^3$ to a domain such as a box with periodic boundary conditions (the torus) or assuming the initial data to be rapidly decaying), with the interaction strength $\lambda$ assumed to be so small that the many small contributions of particle pair interactions sum to an effective external potential (the so--called mean field). The effective potential itself depends on the wave function $\psi$, making the effective description non--linear.

Let us derive the precise choice of parameters. For this argument we restrict attention to the torus, i.\,e., $H$ acting on $L^2(\Tbb^{3N})$, where $\Tbb^3 := \Rbb^3/2\pi\Zbb^3$. The simplest imaginable wave function in the antisymmetric subspace is a Slater determinant of plane waves
\begin{equation}
\psi(x_1,x_2,\ldots x_N) := (N!)^{-1/2} \det(f_j(x_i))\;, \quad \textnormal{ with } f_j(x) := (2\pi)^{-3/2} e^{i k_j \cdot x} \textnormal{ for }x\in \Tbb^3\;.  \label{eq:fermiball}
\end{equation}
If $B_\F := \{ k \in \Zbb^3 : \lvert k\rvert \leq k_\F \}$ for some $k_\F >0$, and $N := \lvert B_\F \rvert$, then the Slater determinant formed by the plane waves $k_j \in B_\F$ is the unique minimizer of the non--interacting Hamiltonian $H= - \sum_{i=1}^N \Delta_{x_i}$. The kinetic energy is then, since $k_\F \sim N^{1/3}$, of the order
\[
\langle \psi, H \psi \rangle = \sum_{k\in B_\F} \lvert k\rvert^2 \sim N^{5/3}\;.
\]
Now let us bring the interaction back into the game.
How small should $\lambda$ be? To have a system in which neither the kinetic energy nor the interaction (as a sum over all pairs typically of order $N^2$) dominates the behavior, we choose
\[
\lambda := N^{-1/3}\;.
\]
Since typical momenta (those close to the ``surface'' of the Fermi ball $B_\F$, and thus the most susceptible to the interaction) are of order $\lvert k\rvert \sim k_\F \sim N^{1/3}$, also the typical velocities of these particles are of order $N^{1/3}$, while the length of the system is $2\pi$. So it is not a severe restriction to look only at short times of order $N^{-1/3}$; rescaling the time variable accordingly, the Schrödinger equation \cref{eq:schroedinger} becomes
\[
 iN^{1/3}\frac{\di \psi}{\di t}(t) =  \left[\sum_{i=1}^N - \Delta_{x_i} + N^{-1/3} \sum_{1 \leq i < j \leq N} V(x_i - x_j) \right] \psi(t)\;.
\]
Introducing the parameter
\[\hbar := N^{-1/3}\]
and multiplying the whole equation by $\hbar^2$, we find a form reminiscent of a naive mean--field scaling limit (having coupling constant $1/N$) and a semiclassical scaling limit (effective Planck constant $\hbar\to 0$):
\begin{equation}\label{eq:coreeq}
 i\hbar\frac{\di \psi}{\di t}(t) =  \left[\sum_{i=1}^N - \hbar^2 \Delta_{x_i} + \frac{1}{N} \sum_{1 \leq i < j \leq N} V(x_i - x_j) \right] \psi(t)\;.
\end{equation}
One expects that the broad idea of the argument is equally applicable, but of course not explicit, for fermions initially placed in a confining potential in $\Rbb^3$ instead of on the torus. Therefore \cref{eq:coreeq} will be the form of the Schrödinger equation I discuss in all of the present review. The scaling presented here was introduced by \cite{NS81,Spo81}.

\paragraph{Reduced Density Matrices}
Associated to $\psi \in L^2_\textnormal{a}(\Rbb^{3N})$ there is the density matrix $\lvert \psi \rangle \langle \psi\rvert$, i.\,e., in Dirac bra--ket notation the projection operator on the subspace spanned by $\psi$. Given a $N$--particle observable $A$, i.\,e., a self--adjoint operator $A$ acting in $L^2_\textnormal{a}(\Rbb^{3N})$, its expectation value may be computed by
\[\langle \psi, A \psi\rangle = \tr_N \big( \lvert \psi\rangle \langle \psi \rvert A \big) \;,\]
the trace being over $L^2_\textnormal{a}(\Rbb^{3N})$. Easier to observe are the averages over all particles of a one--particle observable. That is, if $a$ is a self--adjoint operator acting in $L^2(\Rbb^3)$, and we write $a_j$ for the operator $a$ acting on the $j$--th of $N$ particles (i.\,e., $a_j = \id \otimes \cdots \otimes a \otimes \id \otimes \cdots \otimes \id$), one considers the expectation value
\[\frac{1}{N} \sum_{j=1}^N \langle \psi, a_j \psi\rangle = \langle \psi, a_1 \psi\rangle = \tr_1 \big( \big( \tr_{N-1} \lvert \psi\rangle\langle \psi\rvert \big) a \big)\;;\]
for the first equality we used the antisymmetry \cref{eq:antisymmetry}, and $\tr_{N-1}$ is the partial trace over $N-1$ particles (i.\,e., over $N-1$ tensor factors). The quantity
\[N \tr_{N-1} \lvert \psi\rangle\langle \psi\rvert =: \gamma^{(1)}_\psi\]
(note the normalization factor $N$; in many conventions this is chosen to be $1$ instead) is called the one--particle reduced density matrix of $\psi$; it is an operator acting in the one--particle space $L^2(\Rbb^3)$. In the analysis of many--body quantum problems, the reduced density matrices are often the most natural quantities to study, as the next two sections will confirm. Since $\gamma^{(1)}_\psi$ is a self--adjoint trace class operator, it has a spectral decomposition
\[\gamma^{(1)}_\psi = \sum_{j \in \Nbb} \lambda_j \lvert \varphi_j\rangle \langle \varphi_j \rvert\;, \quad \varphi_j \in L^2(\Rbb^3)\;, \quad \lambda_j \in \Rbb\;.\]
This may be used to define the integral kernel of the one--particle reduced density matrix and in particular its ``diagonal'' (the latter physically corresponding to the density of particles expected at position $x \in \Rbb^3$)
\[\gamma^{(1)}_\psi(x;x') := \sum_{j \in \Nbb} \lambda_j \varphi_j(x)  \overline{\varphi_j(x')}\;, \qquad \gamma^{(1)}_\psi(x;x) := \sum_{j \in \Nbb} \lambda_j \lvert \varphi_j(x)\rvert^2 \;.\]

Assuming that the many--body state is a Slater determinant
\[
\psi(x_1,x_2,\ldots x_N) = (N!)^{-1/2} \det(\varphi_j(x_i))\quad \textnormal{ with arbitrary } \varphi_j \in L^2(\Rbb^3)\;, 
\]
the many--body state and the one--particle reduced density matrix are in one--to--one correspondence (up to multiplication by a phase). In fact, the one--particle reduced density matrix of a Slater determinant is a rank--$N$ projection operator on $L^2(\Rbb^3)$, i.\,e.,
\begin{equation} \label{eq:gamma1}
\gamma^{(1)}_\psi = \sum_{j=1}^N \lvert \varphi_j\rangle \langle \varphi_j \rvert\;.
\end{equation}
Conversely, given a rank--$N$ projection operator, we can compute its spectral decomposition \cref{eq:gamma1} to find the orbitals $\varphi_j$; using the orbitals one can write down the corresponding Slater determinant.

\section{The Semiclassical Theory: Vlasov Equation}
\label{sec:vlasov}
The first level of approximation is provided by semiclassical theory. While the state of the quantum system is described by a vector $\psi \in L^2_\textnormal{a}(\Rbb^{3N})$, a classical system is described by a particle density $f: \Rbb^3 \times \Rbb^3 \to [0,\infty)$ on phase space. This is, $f(x,p)$ describes the fraction of particles which are at position $x \in \Rbb^3$ and have momentum $p \in \Rbb^3$; as a probability density, $f$ should satisfy $f(x,p) \geq 0$ and $\int_{\Rbb^3 \times \Rbb^3} f(x,p) \di x \di p=1$.

\paragraph{Vlasov Equation}  The expected classical evolution equation for $f$ is the Vlasov equation
\begin{equation}\label{eq:vlasov}
\frac{\partial f}{\partial t}(t) + 2 p\cdot \nabla_x f(t) = -F(f(t))\cdot \nabla_v f(t)\;,
\end{equation}
where the mean--field force $F$ is given by $F(f(t)) := - \nabla(V \ast \rho_{f(t)})$, the position space particle density appearing here being $\rho_{f(t)}(x) := \int f(t,x,p) \di p$.

\paragraph{Wigner Function} The key idea of the semiclassical approximation is to associate a function $W_\psi: \Rbb^3 \times \Rbb^3 \to \Rbb$ to a vector $\psi \in L^2_\textnormal{a}(\Rbb^{3N})$. One then assumes $\psi$ to be a solution of the time--dependent Schrödinger equation \cref{eq:coreeq} and considers the evolution of $W_\psi$ in the semiclassical limit of Planck constant $\hbar \to 0$. A common choice is the Wigner function
\begin{equation}\label{eq:wigner}
W_\psi(x,p) := \frac{1}{(2\pi)^3} \int e^{-ip\cdot y/\hbar}\; \gamma^{(1)}_\psi\Big(x + \frac{y}{2}; x - \frac{y}{2}\Big) \; \di y\;.
\end{equation}
Also in the Wigner function we consider $\hbar = N^{-1/3}$. The Wigner function satisfies all the properties of a probability density on phase space, except that it usually has negative parts \cite{SC83,BW95}. The relation between the one--particle density matrix and the Wigner function is inverted by the Weyl quantization:
\begin{equation}
 \gamma^{(1)}_\psi(x;y) = N \int  W_\psi\Big( \frac{x+y}{2},p \Big) e^{ip\cdot(x-y)/\hbar} \di p\;.
\end{equation}

\medskip

The Vlasov equation as an approximation to the fermionic many--body dynamics of pure states is justified by the following theorem.
\begin{thm}[{Vlasov Dynamics, combining \cref{thm:hf} below and \cite[Theorem~2.4]{BPSS16}}]
 Assume that $V \in L^1(\Rbb^3)$ and $\int \lvert \hat{V}(p)\rvert (1+\lvert p\rvert^3) \di p < \infty$. Let $\omega_N$ be a sequence of rank--$N$ projection operators on $L^2(\Rbb^3)$, and assume there exists $C >0$ such that for all $i \in \{1,2,3\}$ the sequence satisfies
\begin{equation} \label{eq:semiclass_vla}
 \lVert [x_i,\omega_N] \rVert_\textnormal{tr} \leq C N \hbar \;, \qquad \lVert [p_i,\omega_N] \rVert_\textnormal{tr} \leq C N \hbar\;,
\end{equation}
where $x_i$ is the position operator and $p_i = -i\hbar\nabla_i$ the momentum operator. Let $\psi_{0}$ be the Slater determinant corresponding to $\omega_N$. 
Assume that we have $W^{1,1}$--regularity uniformly with respect to $N$, i.\,e., there exists $C >0$ such that
\begin{equation} \label{eq:W11}
\norm{W_{\psi_{0}}}_{W^{1,1}} := \sum_{\lvert \beta\rvert \leq 1} \int \lvert \nabla^\beta W_{\psi_{0}}(x,p) \rvert \di x\di p \leq C \;. 
\end{equation}
Let $\gamma^{(1)}(t)$ be the one--particle reduced density matrix associated to the solution of the Schrödinger equation, $\psi(t) := e^{-i H t/\hbar} \psi_{0}$. Let $f(t)$ be the solution of the Vlasov equation with initial data $f(0) := W_{\psi_{0}}$, and $\omega^\textnormal{Vlasov}(t)$ the Weyl quantization of $f(t)$.

Then there exists $C, c_1, c_2 > 0$ such that
\begin{equation}\label{eq:vlasovest}
 \lvert \operatorname{tr} e^{i(\alpha\cdot x + \beta \cdot p)} \Big( \gamma^{(1)}(t) - \omega^{\textnormal{Vlasov}}(t) \Big) \rvert \leq C N\hbar (1+\lvert \alpha\rvert + \lvert \beta\rvert) \exp(c_2 \exp(c_1 \lvert t\rvert))
\end{equation}
for all $\alpha,\beta \in \Rbb^3$ and all $t \in \Rbb$.
\end{thm}
\begin{rems}
 \begin{enumerate} 
 \item Note that $\norm{\gamma^{(1)}(t)}_{\tr} = \norm{\omega^\textnormal{Vlasov}(t)}_{\tr} = N$; the bound \cref{eq:vlasovest} is non--trivial, showing that their difference (at least when tested with the observable $e^{i(\alpha\cdot x + \beta\cdot p)}$, $x$ being the position operator and $p$ the momentum operator) is by $\hbar = N^{-1/3}$ smaller.
  \item There are two lines of proof for the derivation of the Vlasov equation. One may directly take the step from the many--body quantum theory to the Vlasov equation \cite{NS81,Spo81}, or one first derives (as discussed in the next section) the time--dependent Hartree--Fock equation \cref{eq:hf} with bounds uniform in $\hbar$ before taking the limit $\hbar \to 0$ of the solution of the Hartree--Fock equation \cite{BPSS16} (with weaker error estimate also \cite{APPP11a}).
  \item \label{remX} For the latter step, from the Hartree--Fock to the Vlasov equation as $\hbar \to 0$, more singular interaction potentials may be treated when considering mixed states as initial data \cite{Saf20,Saf20a,Saf21,LS21,CLS22}. In that case one has only $0 \leq \omega_N \leq 1$ but not $\omega_N = \omega_N^2$.
  \item Alternatively, convergence of Hartree--Fock solutions to Vlasov solutions with singular interaction potential has also been proved in \cite{LP93,MM93} (without exchange term) and \cite{GIMS98} (for the full Hartree--Fock equation), however only as weak convergence. Explicit bounds using the semiclassical Wasserstein pseudo--distance \cite{GP21} where later obtained by \cite{Laf19,Laf21}.
  \item In \cite{PP09,AKN13,AKN13a} expansions of the solution of the Hartree--Fock equation in powers of $\hbar$, with leading order given by the Vlasov equation, have been derived.
 \end{enumerate}
\end{rems}

\paragraph{Initial Data}
The construction of initial data satisfying all the assumptions is non--trivial. On the one hand, one may use coherent states \cite{BPSS16} (Gaussian wave packets with momentum roughly localized around $p \in \Rbb^3$ and position roughly localized around $r \in \Rbb^3$) of the form
\[f_{r,p}(x) := \hbar^{-3/2} e^{-i p\cdot x/\hbar} \frac{e^{-(x-r)^2/2\delta^2}}{(2\pi \delta^2)^{3/4}} \;, \quad x \in \Rbb^3\,,\ \delta > 0\,,\]
to define with some probability density $M \in W^{1,1}(\Rbb^3 \times \Rbb^3)$ the sequence of density matrices
\[\omega_N(x;y) = \int M(r,p) f_{r,p}(x) \overline{f_{r,p}(y)}\;.\]
 One easily sees that by this construction we satisfy \cref{eq:semiclass_vla} and \cref{eq:W11}, but generally this form of $\omega_N$ is not the one--particle reduced density matrix of a pure $N$--particle state.
 
 On the other hand, if $\omega_N$ is a rank--$N$ projection such as the one--particle reduced density matrix of the ground state of non--interacting fermions in a trapping potential, semiclassical analysis suggests its Wigner transform to be approximately an indicator function in phase space, with accordingly little regularity. A complete understanding of the admissible initial data, and possibly the extension to a larger class, remain interesting problems.
 
 (For mixed states it is easier to construct initial data with regular Wigner function, see the results mentioned in Remark~\ref{remX} above.)
 

\section{The Mean--Field Theory: Hartree--Fock Equation}
The second, more precise, level of approximation is provided by a quantum theory of mean--field type. Unlike the semiclassical theory, this theory is described in terms of a quantum state, i.\,.e., vector in the many--body Hilbert space. The key simplification is that only a submanifold of states with the minimum amount of correlations compatible with the antisymmetry requirement of indistinguishable fermions is considered. Unlike the many--body Schrödinger equation, the effective evolution equation in this submanifold (the Hartree--Fock equation) is non--linear, with the many--body interaction having been replaced by an effective external potential generated by averaging over the position of all other particles.

\paragraph{Hartree--Fock Theory}
The key idea of Hartree--Fock theory is to restrict the quantum many--body problem from $L^2_\textnormal{a}(\Rbb^{3N})$ to the submanifold given by Slater determinants
\begin{equation}\label{eq:slater}
\psi(x_1,x_2,\ldots x_N) = (N!)^{-1/2} \det(\varphi_j(x_i))\;, \quad \textnormal{ with } \varphi_j \in L^2(\Rbb^3)\;. 
\end{equation}
The choice of the orbitals $\varphi_j$ is to be optimized in Hartree--Fock theory. (The restriction compared to the full space $L^2_\textnormal{a}(\Rbb^{3N})$ consists of not permitting linear combinations of Slater determinants.) 
 The time--dependent Schr\"odinger equation for the evolution of $\psi$ can be locally projected onto the tangent space of this submanifold (illustrated in \cref{fig:df});
\begin{figure}
\centering    
 \includegraphics[trim={3.3cm 22.5cm 10cm 3cm},clip,scale=1]{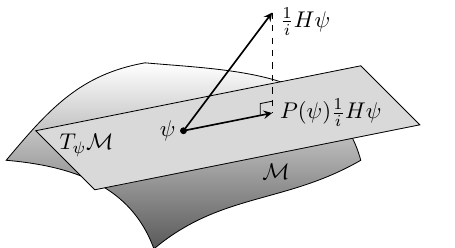}
 \caption{Dirac--Frenkel principle: Consider the Schrödinger equation in a Hilbert space $\Hcal$ and let $\Mcal \subset \Hcal$ be a submanifold. Let $\psi \in \Mcal$. At any ``time step'', $\frac{1}{i}H\psi$ of \cref{eq:schroedinger} is orthogonally projected to the tangent space $T_\psi\Mcal$, yielding an evolution in $\Mcal$. {\footnotesize Figure adapted from \cite{Lub08,BSS18}.}}
 \label{fig:df}
\end{figure}
this gives rise to the system of time--dependent Hartree--Fock equations for the evolution of the orbitals:
\begin{equation}\label{eq:hforb}
i \hbar \frac{\di \varphi_j(t)}{\di t} = - \hbar^2 \Delta \varphi_j(t) + \frac{1}{N}\sum_{i=1}^N \big(V \ast \lvert \varphi_i(t)\rvert^2\big) \varphi_{j}(t) - \frac{1}{N} \sum_{i=1}^N\big(V \ast (\varphi_{j}(t) \overline{\varphi_i(t)}\big) \varphi_i(t)) \;.
\end{equation}
Since Slater determinants are in one--to--one correspondence with their one--particle reduced density matrices, it is natural to write the time--dependent Hartree--Fock equation \cref{eq:hforb} directly in terms of a one--particle density matrix $\omega_{N}(t) := \sum_{j=1}^N \lvert \varphi_{j}(t)\rangle \langle \varphi_{j}(t) \rvert$: 
\begin{align*}
 & i\hbar \frac{\di \omega_{N}(t)}{\di t} = [-\hbar^2\Delta + (V \ast \rho(t)) -X(t), \omega_{N}(t)]\;, \tagg{hf}\\
 & \textnormal{where} \quad \rho(t)(x) := \omega_N(t)(x;x)\;, \quad X(t)(x;x') := V(x - x') \omega_N(t)(x;x')\;.
\end{align*}
The term $V \ast \rho(t)$, a multiplication operator, is called the direct term. The so--called exchange term $X(t)$ is understood with $X(t)(x;x')$ as the integral kernel of an operator.
The Hartree--Fock equation in terms of a one--particle density matrix may also be derived via a reformulation of the Dirac--Frenkel principle for the reduced density matrix \cite{BSS18}. 

\paragraph{Quantum Quench} The typical experimental situation is a quantum quench: a low--energy state (or even the ground state) of fermions in a confining potential is prepared, then by switching the interaction between particles (e.\,g., via a Feshbach resonance) or by switching the confining external potential, the previously prepared state becomes excited with respect to the switched Hamiltonian, thus exhibiting non--trivial dynamics. This dynamics is then observed. The following theorem proves that such a quench can be described by the time--dependent Hartree--Fock equation. To illustrate the idea we only give the simplest case, in which the initial data is exactly a Slater determinant (one may generalize to initial data containing a small number of particles excited over the Slater determinant).           

\begin{thm}[Hartree--Fock Dynamics, \cite{BPS14b,BPS14}]\label{thm:hf}
Let $V \in L^1(\Rbb^3)$ and $\int \di p (1+\lvert p\rvert)^2 \lvert \hat{V}(p)\rvert < \infty$. Let $\omega_N$ be a sequence of rank--$N$ projection operators on $L^2(\Rbb^3)$, and assume there exists $C >0$ such that for all $i \in \{1,2,3\}$ the sequence satisfies
\begin{equation} \label{eq:semiclass}
 \lVert [x_i,\omega_N] \rVert_\textnormal{tr} \leq C N \hbar \;, \qquad \lVert [p_i,\omega_N] \rVert_\textnormal{tr} \leq C N \hbar\;.
\end{equation}
Let $\psi_{0}$ be the Slater determinant corresponding to $\omega_N$. Let $\gamma^{(1)}(t)$ be the one--particle reduced density matrix associated to the solution of the Schrödinger equation $\psi(t) := e^{-i H t/\hbar} \psi_{0}$. If $\omega(t)$ is the solution of the Hartree--Fock equation \cref{eq:hf} with initial data $\omega_N$, then 
\begin{equation}\label{eq:hfest}
 \norm{ \gamma^{(1)}(t) - \omega(t)}_{\operatorname{tr}} \leq C N^{1/6} \exp(c_2 \exp(c_1 \lvert t\rvert)) \qquad \text{for all }t \in \Rbb\;.
\end{equation}
\end{thm}
\begin{rems}
 \begin{enumerate}
 \item Note that $\norm{\gamma^{(1)}(t)}_{\tr} = \norm{\omega(t)}_{\tr} = N$; their difference is by $N^{-5/6}$ smaller.
 \item The exchange term $X(t)$ in \cref{eq:hf} may be dropped without changing the error bound of \cref{eq:hfest}, see \cite[Appendix~A]{BPS14b}.
  \item A similar theorem can be proven with relativistic kinetic energy $\sqrt{-\hbar^2 \Delta + m^2}$ of massive particles, $m > 0$, replacing $-\Delta$ \cite{BPS14a}.
  \item A similar theorem has first been proven by \cite{EESY04}, under assumption of analytic interaction potential, and with error term controllable for short times.
  \item Singular $V$ have been considered in \cite{PRSS17,Saf18}, however only for translation invariant initial data, which are stationary under the Hartree--Fock evolution.
  \item The Hartree--Fock equation has also been derived for initial data given by a mixed state \cite{BJP+16}. This has been generalized to singular interaction potentials, including the Coulomb potential and the gravitational attraction, at least up to small times, in \cite{CLS21}, and generalized by \cite{CLS22a}. Mixed initial states are particularly important in view of the discussion of admissible initial data concerning the derivation of the Vlasov equation in \cref{sec:vlasov}.
  \item The derivation of Hartree--Fock equations has also been considered in scaling limits where the interaction is weaker \cite{BGGM03,BGGM04,FK11,PP16,BBP+16}.
 \end{enumerate}
\end{rems}

\subsection{Initial Data: Non--Interacting Fermions in a Harmonic Trap}
\label{sec:ho}
In \cref{thm:hf} a key role is played by the assumption that the one--particle reduced density matrix of the initial Slater determinant satisfies the semiclassical commutator bounds \cref{eq:semiclass}. The only example given by \cite{BPS14b} was the initial data constituted by the ground state of non--interacting fermions on a torus, i.\,e., a Slater determinant of planes waves \cref{eq:fermiball} whose momenta form a complete Fermi ball
\begin{equation}\label{eq:BF}
B_\F = \{ k \in \Zbb^3 : \lvert k\rvert \leq k_\F \} \quad \textnormal{for some } k_\F >0\;. 
\end{equation}
In \cite{FM20} it was shown that non--interacting fermions in general confining potentials exhibit the semiclassical structure, the proof using methods of semiclassical analysis. Instead in the following we verify \cref{eq:semiclass} by an explicit computation for non--interacting fermions in a harmonic trap.

\medskip

We consider the Hamiltonian $h$, acting on $L^2(\Rbb^3)$, describing a single particle in a three--dimensional anisotropic harmonic oscillator potential
\begin{equation} \label{eq:ho}
h = \sum_{i=1}^3 \left( p_i^2 + w_i^2 x_i^2 \right)\;, \quad \textnormal{with }w_i > 0 \textnormal{ for }i\in \{1,2,3\}\;.
\end{equation}
 We introduce standard creation and annihilation operators by
\begin{equation}
\label{eq:ca-op}
a_i := \sqrt{\frac{w_i}{2\hbar}}\left( x_i + \frac{i}{w_i} p_i \right),\quad a_i^* := \sqrt{\frac{w_i}{2\hbar}} \left( x_i - \frac{i}{w_i} p_i \right)\,.
\end{equation}
Then the Hamiltonian $h$ becomes diagonal, and we can read off its spectrum:
\[h = \hbar \sum_{i=1}^3 2 w_i \left( a_i^* a_i + \frac{1}{2} \right),\ \quad \sigma(h) = \left\{ \hbar \sum_{i=1}^3 2 w_i \left(n_i + \frac{1}{2}\right) : n_i \in \Nbb \right\}\;.\]

Now consider $N$ non--interacting fermions in a harmonic external potential, i.\,e., as an operator acting in $L^2_\textnormal{a}(\Rbb^{3N})$ we consider the Hamiltonian
\[H = \sum_{j=1}^N h_j\;.\]
(In the language of second quantization this is the operator $\di\Gamma(h)$ on the $N$--particle subspace of the fermionic Fock space $\fock$ over $L^2(\Rbb^3)$.) The ground state of $H$ is the antisymmetrized tensor product of the $N$ lowest energy levels of the one--body Hamiltonian $h$, i.\,e., the eigenfunctions associated with the $n_i$ up to a certain $n_i^\textnormal{max}$ form a Slater determinant. To occupy the eigenfunctions from all three oscillators up to the same energy, assuming without loss of generality $w_1 \leq w_2 \leq w_3$, we take $E>0$ and set
\begin{equation} \label{eq:nmax}
n_1^\textnormal{max} := N^{1/3} E\,, \quad n_2^\textnormal{max} := N^{1/3} E \frac{w_1}{w_2}\,, \quad n_3^\textnormal{max} := N^{1/3} E \frac{w_1}{w_3}\,.
\end{equation}
(To be precise we should round to integer values.)
The one--particle reduced density matrix of the corresponding Slater determinant is
\begin{equation}\label{eq:gammaN}
\omega_N = \sum_{n_1,n_2,n_3 = 0}^{n_1^\textnormal{max},n_2^\textnormal{max},n_3^\textnormal{max}} \lvert n_1,n_2,n_3\rangle \langle n_1,n_2,n_3\rvert\;.
\end{equation}
(Here we have introduced the occupation number representation and Dirac bra--ket notation, i.\,e., $\lvert n_1,n_2,n_3\rangle \langle n_1,n_2,n_3\rvert$ denotes the projection on the tensor product of an eigenfunction to eigenvalue $n_1$, an eigenfunction to $n_2$, and an eigenfunction to $n_3$, this triple tensor product forming a wave function in the one--particle space $L^2(\Rbb^3)$.)

According to the following theorem, non--interacting fermions in a harmonic confinement satisfy the semiclassical commutator bounds used to derive the  Hartree--Fock dynamics.
\begin{thm}[Semiclassical Structure of  Non--Interacting Fermions in a Harmonic Trap]\label{prp:nonint}\ 
There is a $C > 0$ such that for all $i \in \{1,2,3\}$ the one--particle density matrix \cref{eq:gammaN} satisfies
\begin{align}\lVert [x_i,\omega_N] \rVert_\textnormal{tr} \leq C N\hbar    \qquad \text{and} \qquad \lVert [p_i,\omega_N] \rVert_\textnormal{tr} \leq C N\hbar \;.    \label{eq:xcomm}
\end{align}
\end{thm}
\begin{proof} We prove the first bound, without loss of generality, for $i=1$. Relation \cref{eq:ca-op} is easily inverted to obtain $x_1 = \sqrt{\hbar/(2 w_1)} \left(a_1 + a_1^* \right)$.
We compute the commutator
\begin{align*}
 [x_1,\omega_N] & = \sqrt{\frac{\hbar}{2 w_1}}\sum_{n_1=0}^{n_1^\textnormal{max}} \Big[a_1^* + a_1,\lvert n_1\rangle\langle n_1\rvert\Big]\otimes \sum_{n_2=0}^{n_2^\textnormal{max}}  \lvert n_2\rangle\langle n_2\rvert \otimes \sum_{n_3=0}^{n_3^\textnormal{max}} \lvert n_3\rangle \langle n_3\rvert\;.
 \end{align*}
By the usual creation operator rules
\begin{align*}& a_1^* \lvert n_1\rangle\langle n_1\rvert - \lvert n_1\rangle\langle n_1\rvert a_1^* + a_1 \lvert n_1\rangle \langle n_1\rvert - \lvert n_1\rangle\langle n_1\rvert a_1\\
 & = \sqrt{n+1} \lvert n_1+1\rangle\langle n_1\rvert - \lvert n_1 \rangle \langle n_1 - 1\rvert \sqrt{n_1} + \lvert n_1-1\rangle\langle n_1\rvert \sqrt{n_1} - \lvert n_1 \rangle \langle n_1 + 1\rvert \sqrt{n_1+1}\,.
\end{align*}
Paying attention to the summation indices (recall that $a_1\lvert 0\rangle = 0$) we find
\begin{align*}
 &\sum_{n_1=0}^{n_1^\textnormal{max}} \sqrt{n_1 + 1}\lvert n_1+1\rangle\langle n_1\rvert - \lvert n_1\rangle\langle n_1+1\rvert \sqrt{n_1 + 1} \\
 & \quad+ \sum_{n_1=1}^{n_1^\textnormal{max}} \lvert n_1-1\rangle \langle n_1\rvert \sqrt{n_1} - \lvert n_1 \rangle \langle n_1 - 1\rvert \sqrt{n_1}\\
 & = \sqrt{n_1^\textnormal{max} + 1} \Big( \lvert n_1^\textnormal{max} + 1 \rangle \langle n_1^\textnormal{max}\rvert - \lvert n_1^\textnormal{max}\rangle \langle n_1^\textnormal{max} + 1\rvert \Big)\,.
\end{align*}
Squaring yields
\begin{align*}
 & \lvert [x_1,\omega_N] \rvert^2 \\
 & =  (n_1^\textnormal{max}+1) \Big( \lvert n_1^\textnormal{max} \rangle \langle n_1^\textnormal{max} +1 \rvert - \lvert n_1^\textnormal{max} +1\rangle \langle n_1^\textnormal{max}\rvert \Big) \Big( \lvert n_1^\textnormal{max} + 1 \rangle \langle n_1^\textnormal{max}\rvert - \lvert n_1^\textnormal{max}\rangle \langle n_1^\textnormal{max} + 1\rvert \Big)\\
 & \quad \otimes \frac{\hbar}{2 w_1}  
 \sum_{n_2=0}^{n_2^\textnormal{max}}  \sum_{\tilde{n}_2=0}^{n_2^\textnormal{max}}  \lvert n_2 \rangle \underbrace{\langle n_2 | \tilde{n}_2\rangle}_{\delta_{n_2,\tilde{n}_2}} \langle \tilde{n}_2\rvert \otimes \sum_{n_3=0}^{n_3^\textnormal{max}} \sum_{\tilde{n}_3=0}^{n_3^\textnormal{max}} \lvert n_3 \rangle \underbrace{\langle n_3 | \tilde{n}_3\rangle}_{\delta_{n_3,\tilde{n}_3}} \langle \tilde{n}_3\rvert 
 \;.
\end{align*}
The square root is easy to calculate since the second and third component of the tensor product are already diagonal and the first one also becomes diagonal when we evaluate the scalar products, leading to
\begin{align*}
 \sqrt{\lvert [x_1,\omega_N] \rvert^2 } & = \sqrt{\frac{\hbar}{2 w_1}} \sqrt{n_1^\textnormal{max} + 1} \Big( \lvert n_1^\textnormal{max} \rangle \langle n_1^\textnormal{max} \rvert + \lvert n_1^\textnormal{max} + 1 \rangle \langle n_1^\textnormal{max} +1 \rvert \Big) \\
 & \quad \otimes \sum_{n_2 = 0}^{n_2^\textnormal{max}}  \lvert n_2 \rangle \langle n_2\rvert \otimes \sum_{n_3=0}^{n_3^\textnormal{max}} \lvert n_3\rangle\langle n_3\rvert\;.
\end{align*}
Finally taking the trace we obtain the claimed bound
\[
 \lVert [x_1,\omega_N] \rVert_\textnormal{tr} = \sqrt{\frac{\hbar}{2 w_1}} \sqrt{n_1^\textnormal{max}+1}\,2\,n_2^\textnormal{max} n_3^\textnormal{max}  =
 \textnormal{const} \times N^{2/3} \sqrt{N^{1/3}E+1}\sqrt{\hbar} \leq CN\hbar\;.
\]

The same holds for the momentum operator because the Hermite functions $\lvert n\rangle$ are eigenvectors of the Fourier transform with the eigenvalues being complex phases, which cancel out from the density matrix; this argument uses that the Fourier transform takes the differential operator $p_1$ into the multiplication operator $x_1$ and by unitarity leaves the trace norm invariant. (Alternatively one can do the calculation analogous to the above also for the momentum operator.)
\end{proof}
This shows that the experimentally most important quantum quench can be described by \cref{thm:hf}: non--interacting fermions are cooled to (almost) temperature $T=0$ in a harmonic trap and then the interaction is switched on and the harmonic confinement switched off.

\begin{rem}
 For the mean--field scaling limit to be non--trivial, the volume should be fixed and the density proportional to total particle number $N$. For \cref{eq:gammaN} one easily computes 
 \begin{align*}
  \tr x_1^2 \omega_N & = \frac{\hbar}{2 w_1} \sum_{n_1=0}^{n_1^\textnormal{max}} \sum_{n_2=0}^{n_2^\textnormal{max}} \sum_{n_3=0}^{n_3^\textnormal{max}} \langle n_1,n_2,n_3 \rvert (a_1 + a_1^*)^2 \lvert n_1,n_2,n_3\rangle  \\
    & = \frac{\hbar}{2 w_1} \sum_{n_1=0}^{n_1^\textnormal{max}} \sum_{n_2=0}^{n_2^\textnormal{max}} \sum_{n_3=0}^{n_3^\textnormal{max}} \langle n_1,n_2,n_3 \rvert 2 a_1^* a_1 + 1 \lvert n_1,n_2,n_3\rangle \\
    & = \frac{\hbar}{2 w_1} \sum_{n_1=0}^{n_i^\textnormal{max}} (2 n_1 + 1)(n_2^\textnormal{max} + 1)(n_3^\textnormal{max}+1)  = \frac{\hbar}{2 w_1} (n_1^\textnormal{max} + 1)^2(n_2^\textnormal{max} + 1)(n_3^\textnormal{max}+1)\;.
 \end{align*}
With $N = \tr \omega_N = (n_1^\textnormal{max} + 1)(n_2^\textnormal{max} + 1)(n_3^\textnormal{max}+1)$ we find the spatial extension
\[\sqrt{\langle x_1^2 \rangle} = \sqrt{\frac{\tr x_1^2 \omega_N}{ \tr \omega_N}} = \sqrt{\frac{\hbar}{2 w_1} (n_1^\textnormal{max} + 1)} = \sqrt{ \frac{E}{2 w_1} + \mathcal{O}(N^{-1/3})} = \Ocal(1) \quad \textnormal{as } N \to \infty\;.\]
So we have indeed $N$ particles in a fixed volume, the density thus being of order $N$ as required.
\end{rem}

\section{Quantum Correlations: Random Phase Approximation}
The random phase approximation (RPA) has originally been introduced by \cite{BP53} for computing the ground state energy to the next order of precision beyond the Hartree--Fock variational approximation. The RPA was later shown to correspond to a formal partial resummation of the perturbative expansion in powers of the interaction \cite{GB57}. A further, morally equivalent formulation of the RPA was developed treating pair excitations as approximately bosonic  particles with a diagonalizable Hamiltonian. This latter ``bosonization'' approach is the only that has found a rigorous justification so far, namely for the ground state energy in \cite{HPR20,CHN21,BNP+20,BNP+21b,BPSS21}. In the following I discuss a recent result showing that the bosonization formulation of the RPA also has a dynamical counterpart, which is valid as a refinement of Hartree--Fock theory to describe the evolution of collective pair excitations over the Fermi ball \cref{eq:BF} of plane waves. The discussion in this section therefore applies only to the case of fermions on the torus $\Tbb^3$. (This is in contrast to the previous sections where particles in $\Rbb^3$ were considered. The restriction to the torus is necessary so that the plane waves are normalizable, and thus can be used as a stationary state of Hartree--Fock theory to which we add the bosonic excitations whose many--body evolution we analyze.)
 
\paragraph{Fock Space Representation}
 To explain the approximate collective bosonization approach developed in \cite{BNP+20}, the method of second quantization is necessary. In second quantization the $N$--particle space $L^2_\textnormal{a}(\Tbb^{3N})$ is embedded in the fermionic Fock space, i.\,e., the direct sum over all possible particle numbers $n$,
 \[\fock := \Cbb \oplus \bigoplus_{n=1}^\infty L^2_\textnormal{a}(\Tbb^{3n})\;.\]
 More explicitly, a vector $\psi_N \in L^2_\textnormal{a}(\Tbb^{3N})$ is identified with a sequence $(0,0,\ldots, 0, \psi_N,0,\ldots) \in \fock$. The advantage of Fock space is that one can use creation and annihilation operators. These are operators on Fock space satisfying the canonical \emph{anticommutator} relations (CAR)
 \[\{a_p,a^*_q\} := a_p a^*_q + a^*_q a_p = \delta_{p,q}\;, \quad \{a_p,a_q\}= 0 = \{a^*_q,a^*_p\} \;, \quad \textnormal{for all momenta } p,q \in \Zbb^3 \,.\]
 We skip the well--known definition of these operators (see, e.\,g., \cite{Sol14}); the convenience of these operators lies exactly in the fact that we only need to know their anticommutators, the fact that applying arbitrary numbers of creation operators $a^*$ to the vacuum vector $\Omega = (1,0,0,\ldots) \in \fock$ one obtains a basis of Fock space $\fock$, and the fact that $\Omega$ lies in the null space of all annihilation operators $a$. The starting point for all further steps is that the Hamiltonian $H$ is just the restriction to the $N$--particle sector of the Fock space Hamiltonian
 \[\Hcal := \hbar^2 \sum_{k \in \Zbb^3} \lvert p\rvert^2 a^*_p a_p + \frac{1}{2N}\sum_{k,q,p \in \Zbb^3} \hat{V}(k) a^*_{p+k} a^*_{q-k} a_q a_p \;.\]
 
 \paragraph{Particle--Hole Transformation}
 In the first step, a particle--hole transformation is used to separate the fixed Fermi ball of plane waves from its excitations. The particle--hole transformation is a unitary map $R: \fock \to \fock$, defined by its properties
 \[R^* a^*_p R := \begin{cases} a_p & \textnormal{for }p\in B_\F \\ a^*_p & \textnormal{for }p \in B_\F^c := \Zbb^3 \setminus B_\F \end{cases}\;, \qquad R\Omega := \psi_{B_\F}\;,\]
the latter vector being the Slater determinant constructed from the plane waves in $B_\F$, as in \cref{eq:fermiball}. Using this rule for conjugation with $R$ and the CAR to arrange the result in normal--order (creation operators $a^*$ to the left of annihilation operators $a$) one obtains
\[R^* \Hcal R = E_\textnormal{HF} + \Hbb_0 + Q_\B + \Ecal\;.\]
The first summand $E_\textnormal{HF} = \langle \psi_{B_\F}, H \psi_{B_\F} \rangle$ is a real number and can be identified as the Hartree--Fock energy. The term
\begin{align*}
 \Hbb_0 & := \sum_{k \in \Zbb^3} e(k) a^*_k a_k\;, \quad e(k) := \hbar^2 \left\lvert  \lvert k\rvert^2 - k_\F^2 \right\rvert\;,        \tagg{disprel}
 \intertext{ is the kinetic energy of pair excitations (removing one momentum mode from inside the Fermi ball by applying an annihilation operator and adding a particle outside the Fermi ball by applying a creation operator). The term}
 Q_\B &:= \frac{1}{N}\sum_{k \in \north} \hat{V}(k) \Big( b^*(k) b(k) + b^*(-k) b(-k) + b^*(k) b^*(-k) +  b(-k) b(k) \Big)        \tagg{QB}
\end{align*} 
is the part of the interaction that can be written purely in terms of \emph{particle--hole excitations} ``delocalized'' over the entire Fermi ball, i.\,e., the linear combinations
\begin{equation} b^*(k) := \sum_{p \in B_\F^c, h\in B_\F} \delta_{p-h,k} a^*_p a^*_h\;.\label{eq:bstark}\end{equation}
The purpose of introducing a separation of the support of $\hat{V}$ into two parts, $\supp \hat{V} = \north \cup (-\north)$, defined by
\[ \north := \{  k \in \Zbb^3 \cap \supp\hat{V} : k_3 > 0 \textnormal{ or ($k_3 =0$ and $k_2 > 0$) or ($k_2 = k_3 =0$ and $k_1 >0$)} \} \;,\]
is that the pair creation operators appear only once in the summand, not both for $k$ and $-k$ (which will permit us to approximate them as independent bosonic modes later). 

All further contributions to the Hamiltonian, i.\,e., everything that is not part of $\Hbb_0$ or cannot be written using the $b^*$-- and $b$--operators, are collected in $\Ecal$ and can be proven to constitute only small error terms, at least when acting on states with few excitations.

\medskip 

At this point the main observation is that $Q_\B$ is quadratic when expressed through the $b^*$-- and $b$--operators; moreover, being (sums of) pairs of anticommuting operators, the $b^*$ among them commute, i.\,e.,
\[ [b^*(k),b^*(l)] := b^*(k) b^*(l) - b^*(l) b^*(k) = 0\;, \]
i.\,e., these operators commute just like bosonic operators. Moreover, the vacuum $\Omega$ is in the null space of all the $b$--operators, just like a vacuum vector in Fock space. One may therefore conjecture that the $b$--operators realize a representation of the canonical \emph{commutator} relations (CCR), which describe bosonic particles in a symmetric Fock space. Recall that true bosonic operators $b_k$ and $b^*_l$ would satisfy the exact CCR
\begin{equation} [b^*_k, b^*_l] = 0 = [b_k,b_l]\;, \qquad [b_k,b^*_l] = \delta_{k,l}\;. \label{eq:exactccr}\end{equation}
That this cannot be quite true is easily noted: whereas by antisymmetry one can never create more than two fermions in the same state (one has $(a^*_k)^2 = 0$, the Pauli exclusion principle), for bosons arbitrary powers of $b^*_k$ never vanish. Since the concrete $b^*(k)$ of \cref{eq:bstark} are constructed from fermionic operators, they will at high powers eventually vanish and thus violate this bosonic property. But as long as we consider states with few excitations, \cref{eq:exactccr} may constitute a valid approximation for the commutator relations of the constructed operators. This will in fact be quantified by \cref{eq:approxccr} below.

For the moment, let us focus on another difficulty: the operator $\Hbb_0$ is not given in terms of $b^*$-- and $b$--operators. To obtain an exactly solvable quantum theory, we need to express not only $Q_\B$ but also $\Hbb_0$ as a quadratic expression in terms of approximately bosonic operators. This will be achieved by the patch decomposition we introduce next. 

\paragraph{Patch Decomposition of the Fermi Surface}
It turns out that a formula for $\Hbb_0$ that is quadratic in the $b^*$-- and $b$--operators can be obtained if the dispersion relation $e(k)$ (as defined in \cref{eq:disprel}) is linearized. To linearize $e(k)$, we argue that all momenta $p \in \BFc$ and $h \in \BF$ belong to a shell around the Fermi surface $\{k \in \Rbb^3: \lvert k\rvert = k_\F\}$. We can then cut this shell into patches and linearize $e(k)$ around the patch centers. So why are all momenta restricted to such a shell? Note that the main term $Q_\B$ of the interaction contains only pair operators in which $k \in \supp \hat{V}$. So assuming $\supp\hat{V}$ to be compact, the pair operators \cref{eq:bstark} because of the requirement $p-h=k$ indeed contain only $p$ and $h$ belonging to a shell of width $\operatorname{diam}(\supp\hat{V})$ around the Fermi surface. The ``northern'' half of this shell (with ``north'' fixed as an arbitrary direction) may then be sliced into patches $B_\alpha$ (with indices $\alpha = 1,\ldots,M/2$) as indicated in \cref{fig:patches}, and this slicing reflected by the origin to the southern half. The total number of patches will be chosen as $M := N^\alpha$ with $\alpha \in (0,2/3)$ (further requirements of the proof narrow down this interval). 
These patches are separated by corridors of width strictly larger than $2\operatorname{diam}\supp\hat{V}$; moreover
they do not degenerate as $N \to \infty$ in the sense that their circumference will always be of order $N^{1/3}/\sqrt{M}$ while they cover an area of size $N^{2/3}/M$ on the Fermi sphere. For every patch $B_\alpha$ we choose a vector $\omega_\alpha$ with $\lvert \omega_\alpha \rvert = k_\F$ near the patch center. 

The main idea is now to localize the pair creation operators to these patches, defining
\begin{equation} b^*_\alpha(k) := \frac{1}{n_\alpha(k)}\sum_{\substack{p \in B_\F^c \cap B_\alpha \\ h\in B_\F \cap B_\alpha}} \delta_{p-h,k} a^*_p a^*_h\;,\label{eq:bstarkalpha}\end{equation}
where we introduced the normalization constant $n_\alpha(k)$ such that $\norm{b^*_\alpha(k) \Omega} = 1$. Here we notice a small problem: only if $k$ points outward the Fermi ball (from a hole momentum $h\in B_\F$ to a particle momentum $p \in B_\F^c$) this sum will be non--zero. If $k$ points radially inward or outward but under a very flat angle to the tangent plane, the sum may be empty or contain very few $(h,p)$--pairs. We therefore impose a cut--off on the set of $\alpha$ such that we keep only those with
\begin{equation}
\label{eq:delta}
k \cdot \frac{\omega_\alpha}{\lvert \omega_\alpha \rvert} \geq N^{-\delta}
\end{equation}
(the choice of $\delta > 0$ may be optimized).
One finds
\[n_\alpha(k) = \bigg[\sum_{\substack{p \in B_\F^c \cap B_\alpha \\ h\in B_\F \cap B_\alpha}} \delta_{p-h,k}\bigg]^{1/2} \approx \sqrt{\frac{4\pi k_\F^2}{M} \left\lvert k\cdot \frac{\omega_\alpha}{\lvert \omega_\alpha \rvert} \right\rvert} \;.\]
 We can now prove that these operators are almost bosonic, in the sense that    
\begin{equation}
 [b^*_\alpha(k),b^*_\beta(l)] = [b_\alpha(k), b_\beta(l)] = 0\;, \quad [b_\alpha(k),b^*_\beta(l)] = \delta_{\alpha,\beta} \left( \delta_{k,l} + \mathcal{\Ecal_\alpha}(k,l) \right) \label{eq:approxccr}
\end{equation}
where the error term of the last commutator can be estimated, e.\,g., by bounds such as $\norm{\Ecal_\alpha(k,l) \psi} \leq 2{n_\alpha(k)^{-1} n_\alpha(l)^{-1}} \norm{\Ncal \psi}$ for all $\psi\in \fock$. Thanks to the introduction of the cut--off and the assumption $M \ll N^{2/3}$ we have $n_\alpha(k) \to \infty$ as $N\to \infty$.

As we have seen, for at least half of the values of $\alpha$, the operators $b^*_\alpha(k)$ vanish. To simplify notation we introduce
\[c^*_\alpha(k) := \begin{cases} b^*_\alpha(k) & \textnormal{ for } \alpha \textnormal{ such that } k \cdot \omega_\alpha / \lvert \omega_\alpha \rvert \geq N^{-\delta}\;, \\
 b^*_\alpha(-k) & \textnormal{ for } \alpha \textnormal{ such that } k \cdot \omega_\alpha / \lvert \omega_\alpha \rvert \leq - N^{-\delta}\;.
\end{cases}\]
(In the following we will use $\Ical_k^{+} := \{\alpha : k \cdot \omega_\alpha / \lvert \omega_\alpha \rvert \geq N^{-\delta} \} \subset \{1,2,\ldots,M\} $ and $\Ical_k := \Ical_k^{+} \cup \Ical_{-k}^{+}$, implicitly depending on the choice of $\delta$.)

Now turning back to the kinetic energy, one may linearize the dispersion relation as claimed around the patch centers $\omega_\alpha$: in fact (without loss of generality for $\alpha \in \Ikp$)
\begin{align*}
 [\Hbb_0, c^*_\alpha(k)] & = \Big[\sum_{i \in \Zbb^3} e(i) a^*_i a_i, \frac{1}{n_\alpha(k)} \sum_{\substack{p \in B_\F^c \cap B_\alpha \\ h\in B_\F \cap B_\alpha}} \delta_{p-h,k} a^*_p a^*_h\Big] \\
 & = \frac{1}{n_\alpha(k)} \sum_{\substack{p \in B_\F^c \cap B_\alpha \\ h\in B_\F \cap B_\alpha}} \delta_{p-h,k} \left( e(p) - e(h) \right) a^*_p a^*_h 
 \approx 2 \hbar^2 k_\F k \cdot \frac{\omega_\alpha}{ \lvert \omega_\alpha \rvert } c^*_\alpha(k) \;.
\end{align*}
The same commutator is obtained replacing $\Hbb_0$ in this formula by
\[\Dbb_\B := 2\hbar^2 k_\F \sum_{k \in \north} \sum_{\alpha \in \Ik} \left\lvert k \cdot \frac{\omega_\alpha}{\lvert \omega_\alpha \rvert} \right\rvert c^*_\alpha(k) c_\alpha(k)\;.\] 
If vectors of the form $\prod_{j=1}^m c^*_{\alpha_j}(k_j) \Omega$, $m \in \Nbb$, constituted a basis of the \emph{fermionic} Fock space, this would imply an identity between the operators $\Hbb_0$ and $\Dbb_\B$. In \cite{BNP+21b,BPSS21} much effort is dedicated to justifying this at least as an approximation of vectors close to the ground state. As far as the approximation of the time evolution presented below is concerned, this will be much less of a problem since we only consider initial data created by the application of the pair creation operators $c^*_\alpha(k)$.
\begin{figure}[t]  
 \centering    
 \includegraphics[trim={5cm 20cm 12cm 3.25cm},clip,scale=1.3]{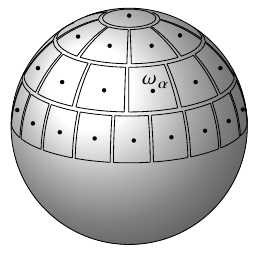}
\caption{The shell around the Fermi surface is decomposed into $M = N^{\alpha}$ patches, with $\alpha > 0$ to be optimized. Patches on the southern half are obtained through reflection at the origin. {\footnotesize (From \cite{BNP+20} under CC BY 4.0 license, \url{http://creativecommons.org/licenses/by/4.0/}, with $\omega_\alpha$ added.)}}  
\label{fig:patches} 
\end{figure}      

\paragraph{Approximately Bosonic Effective Hamiltonian}
We may now combine what we learned about the dominant interaction term $Q_\B$ and the kinetic energy to state the bosonic theory providing us with the effective evolution of particle--hole pair excitations. Summing the approximate kinetic energy $\Dbb_\B$ and the dominant interaction terms $Q_\B$, and decomposing 
\[b^*(k) \approx \sum_{\alpha \in \Ikp} n_\alpha(k) c^*_\alpha(k) \;,\]
we find the approximation (with $\kappa := (3/4\pi)^{1/3}$)
\begin{equation}\label{eq:effbosH}
\Hcal_\textnormal{corr}= R^* \Hcal_N R - E^{\textnormal{pw}}_{N} \approx \sum_{k\in \north} 2\hbar \kappaf \lvert k\rvert h_{\textnormal{eff}}(k)
\end{equation}
with the effective Hamiltonian
\begin{align}\label{eq:heff}
h_\textnormal{eff}(k)  := \!\sum_{\alpha,\beta \in \Ik} \Big[ \big( \D(k) + \W(k) \big)_{\alpha,\beta} c_\alpha^*(k) c_\beta(k) +  \frac{1}{2} \Wt(k)_{\alpha,\beta} \big( c^*_\alpha(k) c^*_\beta(k) + \hc \big)  \Big]
\end{align}
where $\D(k)$, $\W(k)$, and $\Wt(k)$ are  real symmetric matrices
\begin{equation}
\label{eq:blocks1}\begin{split}
\D (k)_{\alpha,\beta} &: =  \delta_{\alpha,\beta} \lvert k \cdot {\omega}_\alpha\rvert / ( \lvert k\rvert \lvert \omega_\alpha\rvert )\;, \qquad \forall \alpha,\beta \in \Ik\;, \\
\W(k)_{\alpha,\beta}  & := \frac{\hat V(k)}{2\hbar \kappaf N \lvert k\rvert} \times \left\{ \begin{array}{cl} n_\alpha(k) n_\beta(k) & \text{ if } \alpha,\beta \in \Ikp \text{ or } \alpha,\beta \in \Ical_{-k}^{+} \\
 0 & \text{ otherwise}\,,\end{array} \right. \\
\Wt(k)_{\alpha,\beta} & := \frac{\hat V(k)}{2\hbar \kappaf N \lvert k\rvert} \times \left\{ \begin{array}{cl} 
0  & \text{ if } \alpha,\beta \in \Ikp \text{ or } \alpha,\beta \in \Ical_{-k}^{+} 
\\ n_\alpha(k) n_\beta(k)  & \text{otherwise}\,.
\end{array} \right.
\end{split}
\end{equation}
The rigorous justification of this (approximately) bosonic Hamiltonian as an approximation to the microscopic fermionic Schr\"odinger equation is provided by \cref{thm:mainRPA} below. To state it we need to discuss the solution (i.\,e., Fock space diagonalization) of the effective Hamiltonian first.

\paragraph{Diagonalization}
If $c_\alpha^*(k)$ were exactly bosonic creation operators, then the quadratic Hamiltonian $h_\textnormal{eff}(k)$ could be diagonalized by a Bogoliubov transformation (a linear automorphism of the CCR algebra) of the form \cite[Appendix A.1]{BNP+20}
\begin{equation}    \label{eq:B}
T:= e^{B}\;, \quad B: = \sum_{k\in \north}   \frac{1}{2} \sum_{\alpha, \beta \in \mathcal{I}_{k}} K(k)_{\alpha, \beta} c^{*}_{\alpha}(k) c^{*}_{\beta}(k) - \hc
\end{equation}
where
\begin{align*} 
K(k) & := \log \lvert S (k)^\intercal \rvert =\frac{1}{2} \log  \Big( S(k) S(k)^\intercal \Big) \;,\\
S(k) & := (\D(k)+\W(k)-\Wt(k))^{\! 1/2} E(k)^{-1/2}\;,  \\
E(k) & := \!\!\Big[\! \big(\D(k)+\W(k)-\Wt(k)\big)^{1/2}\! (\D(k)+\W(k)+\Wt(k)) \big(\D(k)+\W(k)-\Wt(k)\big)^{\! 1/2} \Big]^{\! 1/2}\!\!. 
\end{align*}
Since our pair operators do not quite satisfy the commutator relations of the CCR algebra, $T$ turns out to be only approximately a Bogoliubov transformation:
\[
T^* c_{\gamma}(k) T \approx \sum_{\alpha\in \mathcal{I}_{k}} \cosh(   K(k) )_{\alpha, \gamma} c_{\alpha}(k) + \sum_{\alpha \in \mathcal{I}_{k}} \sinh (  K(k) )_{\alpha, \gamma} c^{*}_{\alpha}(k) \;.
\]
With the indicated choice of $K(k)$, the ``off--diagonal'' terms in the quadratic Hamiltonian (i.\,e., those of the form $c^* c^*$ and $c c$) are approximately cancelled by conjugation with the unitary $T$ (see the proof of \cite[Lemma~10.1]{BNP+21b}), so that
\begin{equation} \label{eq:corr-diag}
T^{*} \Hcal_\textnormal{corr}  T  \approx  \widetilde E_N^{\textnormal{RPA}} + \sum_{k\in \north} 2\hbar \kappaf \lvert k\rvert \sum_{\alpha, \beta \in \mathcal{I}_{k}} \mathfrak{K}(k)_{\alpha, \beta} c^{*}_{\alpha}(k) c_{\beta}(k)
\end{equation} 
with the Hermitian matrix
\begin{align}
\mathfrak{K}(k) &= \cosh(K(k)) (D(k) + W(k)) \cosh(K(k)) + \sinh(K(k)) (D(k) + W(k)) \sinh(K(k)) \nonumber\\
& \quad + \cosh(K(k)) \widetilde{W}(k) \sinh(K(k)) + \sinh(K(k)) \widetilde{W}(k) \cosh(K(k)) \label{eq:curlyKexplicit}
\end{align}
and the RPA prediction for the ground state energy correction
\begin{align}  \label{eq:GSE-RPA}
\widetilde E_N^{\textnormal{RPA}}=  \sum_{k\in \north} \hbar \kappaf \lvert k\rvert \tr (E(k) - \D(k)-\W(k)) \in \Rbb\;. 
\end{align}
Thus $h_\textnormal{eff}(k)$ can be understood as the approximately bosonic second quantization of the operator $\mathfrak{K}(k)$ on the one--boson space $\ell^2(\Ik) \simeq \Cbb^{\lvert \Ik\rvert}$. If the effective Hamiltonians at different momenta $k$ were independent, we could simply sum over $k\in \north$ and find that the excitation spectrum consists of sums of eigenvalues of $2\hbar \kappaf \lvert k\rvert E(k)$ (see \cite{Ben21}). 

\paragraph{Particle--Hole Pairs: Initial Data and Bosonic Dynamics}
The theorem will describe the evolution of collective particle--hole excitations of the Fermi ball. We consider the many--body Schr\"odinger equation with the initial data
\begin{equation} \label{eq:xi-def-ini}
\psi := R T \xi  \in L^2_a(\Rbb^{3N})\;,\qquad  \xi := \frac{1}{Z_{m}} c^{*}(\varphi_{1}) \cdots c^{*}(\varphi_{m}) \Omega\;,
\end{equation}
where
\begin{equation}
c^{*}(\varphi_{i}) = \sum_{k \in \north} \sum_{\alpha \in \mathcal{I}_{k}} c^{*}_{\alpha}(k) (\varphi_{i}(k))_{\alpha}
\end{equation}
with a number $m\in \Nbb$ of one--boson states 
\begin{equation}\label{eq:139}
\varphi_{1}, \ldots, \varphi_{m} \in \bigoplus_{k\in \north} \ell^2(\mathcal{I}_k) \;, \quad 
\norm{ \varphi_{i} }^2 := \sum_{k\in \north} \sum_{\alpha \in \mathcal{I}_{k}} \lvert (\varphi_{i}(k))_{\alpha}\rvert^2 = 1\;.
\end{equation}
We do not require orthogonality of the functions $\varphi_i$: since they describe approximately bosonic excitations, they may all occupy the same one--boson state $\varphi_1$.
The normalization constant $Z_{m}$ is chosen such that $\| \xi \| = 1$.  
We will approximate the evolution of such initial data using the effective evolution
\begin{equation} \label{eq:def-xi-t}
\xi(t) := \frac{1}{Z_{m}} c^{*}(\varphi_{1}(t)) \cdots c^{*}(\varphi_{m}(t)) \Omega\;,\qquad t\in \Rbb\;,
\end{equation}
where, with $\mathfrak{K}(k)$ defined in \eqref{eq:corr-diag}, 
\begin{equation}\label{eq:evphi}
\varphi_{m}(t) := e^{-i H_{\textnormal{B}} t/\hbar} \varphi_{m}\;,\qquad H_{\textnormal{B}} := \bigoplus_{k\in \north} 2 \hbar \kappaf |k| \mathfrak{K}(k)\;.
\end{equation}
The state $\xi(t)$ can be viewed as an approximate $m$--boson state, where every $\varphi_{i}$ evolves independently according to the one--boson Hamiltonian $H_{\textnormal{B}}$. We can now state the theorem.

\begin{thm}[RPA Dynamics, \cite{BNP+22}]  
\label{thm:mainRPA}
Assume that $\hat{V}: \Zbb^3 \to \Rbb$ is compactly supported, non--negative, and $\hat{V}(k)=\hat{V}(-k)$ for all $k\in \Zbb^3$. Let $k_\F >0$, $N := \lvert\{ k\in \Zbb^3: \lvert k\rvert \leq k_\F \}\rvert$, and $\hbar := \kappaf k_\F^{-1} = N^{-1/3} + \Ocal(N^{-2/3})$ with $\kappaf = \left(3/4\pi\right)^{\frac{1}{3}}$. Moreover assume that the number $m$ of pair excitations satisfies $m^3 (2m-1)!! \ll N^\delta$, where $\delta$ is given by \cref{eq:delta}. Then there exists a $C_{V} > 0$ such that for any $t \in \Rbb$ we have
\begin{align}  \label{eq:main-1-weaker-bb}
& \norm{ e^{-i\mathcal{H} t/\hbar} R T \xi - e^{-i (E^{\textnormal{pw}}_{N} + \widetilde{E}_N^\textnormal{RPA})t/\hbar} R T \xi(t) } \\
& \leq C_V  (m+1)^2\sqrt{(2m-1)!!} \Big( N^{-\frac{\delta}{2}}  + M^{-\frac{1}{2}} + M^{\frac{3}{2}}N^{-\frac{1}{3} + \delta} + M^{\frac{1}{4}} N^{-\frac{1}{6}}  \Big)|t| \;. 
\end{align}
\end{thm}
\begin{rems}
 \begin{enumerate}
 \item The states $RT\xi$ and $RT\xi(t)$ are $N$--particle states. In particular the action of the second quantized $\Hcal$ agrees with the action of $H$ on these states. This follows since the pair operators $c^*_\alpha(k)$ create equal numbers of particles $p \in \BFc$ and holes $h \in \BF$. More precisely, with $\Ncal^\textnormal{p} := \sum_{p \in \BFc} a^*_p a_p$ and $ \Ncal^\textnormal{h} := \sum_{h \in \BF} a^*_h a_h$ one has $(\Ncal^\textnormal{p} - \Ncal^\textnormal{h}) T\xi = 0$, which implies
 \[\Ncal RT \xi = R(R^*\Ncal R)T\xi = R(\Ncal^\textnormal{p} - \Ncal^\textnormal{h} + N)T\xi = R N T\xi = NRT\xi \;. \]
 \item In \cite{BNP+22} we specialized to the number of patches $M:=N^{4\delta}$ and the cut--off parameter $\delta :=2/45$ (entering in \cref{eq:delta}). The present form is more general since it also describes the evolution of initial data given with non--optimal choice of $M$ and $\delta$. (But of course the theorem is only of interest when the error estimate is $\ll 1$.)
 \item We avoided some trivial contributions to the error by keeping $\widetilde{E}^\textnormal{RPA}_N$ instead of replacing it by a (more explicit) integral formula as in \cite{BNP+22}.
 \item The theorem not only provides a stronger approximation than \cref{thm:hf} by employing Fock space norm instead of the trace norm of a reduced density matrix, but also has a better time dependence of the error.
 \item The mentioned improvement comes at a cost: the theorem is only applicable to initial data given in terms of pair excitations over the Fermi ball. According to \cite[Appendix~A]{BNP+21b}, the Fermi ball constitutes the minimizer (due to the scaling limit, in general only a stationary point) of the Hartree--Fock variational problem (i.\,e., minimization of $\langle \psi, H \psi\rangle$ over Slater determinants $\psi$ on the torus) and is thus stationary for the time--dependent Hartree--Fock equation \cref{eq:hf}. The theorem  does not apply, e.\,g., to the harmonically confined Fermi gas, where we reach only the Hartree--Fock precision of the previous section. 
 \item Only the pair excitations have a non--trivial evolution, the Fermi ball remains stationary. The spectrum of pair excitations has been discussed in \cite{Ben21,CHN21}. 
 \item Different bosonization concepts appear in the analysis of low--density Fermi gases \cite{FGHP21}, spin systems \cite{CG12,CGS15,Ben17,NS21a}, and one--dimensional fermionic systems \cite{ML65,LLMM17,LLMM17a}.
 \end{enumerate}   
\end{rems}
  
\paragraph{Concluding Remarks}
I have described three levels of approximation for the dynamics of the fermionic many--body problem at high densities. While providing increasingly precise results (from approximation of the Wigner transform to approximation of reduced density matrices in trace norm to a Fock space norm approximation) we have also seen the role of the initial conditions, such as regularity of the Wigner transform when deriving the Vlasov equation, semiclassical commutator bounds for the validity of Hartree--Fock theory, and the initial data consisting of pair excitations over a stationary Fermi ball for the RPA. Moreover we have seen that the generalization of these assumptions still provides a number of important questions on which further progress would be desirable.
 
\section*{Acknowledgements and Declarations}
The author has been supported by the Gruppo Nazionale per la Fisica Matematica (GNFM) of the Istituto Nazionale di Alta Matematica ``Francesco Severi'' (INdAM) in Italy and the European Research Council (ERC) through the Starting Grant \textsc{FermiMath}, grant agreement nr.~101040991. The author does not have any conflicts of interest to disclose. Data sharing is not applicable to this article as no new data were created or analyzed.

\bibliographystyle{alpha}  
\bibliography{semiclassicalstructure}{}   

\end{document}

%% file: macros.tex
\newtheorem{thm}{Theorem}[section]

\theoremstyle{remark}
\newtheorem*{rem}{Remark}
\newtheorem*{rems}{Remarks}

\renewcommand{\theequation}{\thesection.\arabic{equation}}

\numberwithin{equation}{section}

\newcommand{\fock}{\mathcal{F}}		
\newcommand{\di}{{\textnormal{d}}}		

\newcommand{\Mcal}{\mathcal{M}}

\newcommand{\Tbb}{\mathbb{T}}
\newcommand{\W}{W}
\newcommand{\Wt}{\widetilde{W}}
\newcommand{\D}{D}

\newcommand{\Ecal}{\mathcal{E}}
\newcommand{\Hbb}{\mathbb{H}}

\newcommand{\Ncal}{\mathcal{N}}		
\newcommand{\Hcal}{\mathcal{H}}		
\newcommand{\Ical}{\mathcal{I}}

\newcommand{\Ikp}{\Ical_{k}^{+}}

\newcommand{\Ik}{\Ical_{k}}

\newcommand{\Ocal}{\mathcal{O}}		
\newcommand{\hc}{\textnormal{h.c.}}		

\newcommand{\Rbb}{\mathbb{R}}		
\newcommand{\Cbb}{\mathbb{C}}		
\newcommand{\Nbb}{\mathbb{N}}		
\newcommand{\Zbb}{\mathbb{Z}}

\newcommand{\Dbb}{\mathbb{D}}

\newcommand{\id}{\mathbb{I}}
\newcommand{\norm}[1]{\lVert#1\rVert}	

\newcommand{\tr}{\operatorname{tr}}

\newcommand{\sgn}{\operatorname{sgn}}
\newcommand{\tagg}[1]{ \stepcounter{equation} \tag{\theequation} \label{eq:#1} } 

\newcommand{\BFc}{B_\F^c}
\newcommand{\BF}{B_\F}

\newcommand{\kappaf}{\kappa}

\newcommand{\north}{\Gamma^{\textnormal{nor}}}

\newcommand{\F}{\textnormal{F}} 
\newcommand{\B}{\textnormal{B}} 

\newcommand{\supp}{\operatorname{supp}}
